\documentclass[10pt]{article} 

\def \onepagefig{1}

\usepackage[a4paper,verbose]{geometry}

\newcommand{\cldiam}{{\tt CL-DIAM}\xspace}

\usepackage{amsmath} 
\usepackage{amsthm}  
\usepackage{amssymb} 
\usepackage[ruled,vlined]{algorithm2e} 
\usepackage{booktabs} 
\usepackage{graphicx}
\usepackage{subcaption}

\newtheorem{definition}{Definition}

\newtheorem{theorem}{Theorem}
\newtheorem{lemma}{Lemma}

\newtheorem{fact}{Fact}
\newtheorem{corollary}{Corollary}

\newcommand{\ord}[1]{$\cdot 10^{#1}$}

\newcommand{\ceil}[1]{\left\lceil #1 \right\rceil}

\newcommand{\BO}[1]{O\left( #1 \right)}       
\newcommand{\BT}[1]{{\Theta}\left(#1\right)}  
\newcommand{\BOM}[1]{{\Omega}\left(#1\right)} 

\newcommand{\dist}{\operatorname{dist}} 

\newcommand{\Let}[2]{#1 $\leftarrow$ #2} 

\SetKw{To}{to}                             
\def\Dend{\Delta_{\mbox{\scriptsize end}}} 
\def\MR{MR$(M_T,M_L)$}                     

\def\ClustR{R_{\mbox{\scriptsize\tt CL}} (\tau)}
\def\Clust2R{R_{\mbox{\scriptsize\tt CL2}} (\tau)}

\newcommand{\dstepping}{$\Delta$-stepping\xspace}
\newcommand{\affaddr}[1]{{\small #1}}  
\newcommand{\email}[1]{{\tt\small #1}} 

\begin{document}

\pagestyle{plain}  
\title{A Practical Parallel Algorithm for Diameter Approximation of Massive Weighted Graphs}

\author{
Matteo Ceccarello$^1$ \and
Andrea Pietracaprina$^1$ \and
Geppino Pucci$^1$ \and
Eli Upfal$^2$
\\
$^1$\affaddr{Department of Information Engineering, University of Padova,
Padova, Italy}\\
\email{\{ceccarel,capri,geppo\}@dei.unipd.it}
\\
$^2$\affaddr{Department of Computer Science,
Brown University, Providence, RI USA}\\
\email{eli\_upfal@brown.edu}
}

\date{}

\maketitle              

\begin{abstract}
We present a space and time efficient practical parallel algorithm for
approximating the diameter of massive weighted undirected graphs on
distributed platforms supporting a MapReduce-like abstraction.  The
core of the algorithm is a weighted graph decomposition strategy
generating disjoint clusters of bounded weighted
radius. Theoretically, our algorithm uses linear space and yields a
polylogarithmic approximation guarantee; moreover, for important practical
classes of graphs, it runs in a number of rounds asymptotically
smaller than those required by the natural approximation provided by
the state-of-the-art $\Delta$-stepping SSSP algorithm, which is its
only practical linear-space competitor in the aforementioned
computational scenario. We complement our theoretical findings with an
extensive experimental analysis on large benchmark graphs, which
demonstrates that our algorithm attains substantial improvements on a
number of key performance indicators with respect to the
aforementioned competitor, while featuring a similar approximation
ratio (a small constant less than 1.4, as opposed to the
polylogarithmic theoretical bound).

\vspace{1em}

\end{abstract}

\section{Introduction}\label{sec-intro}
The analysis of very large (typically sparse) weighted graphs is
becoming a central tool in numerous domains, including geographic
information systems, social sciences, computational biology, computational
linguistics, semantic search and knowledge discovery, and cyber
security. A fundamental primitive for graph analytics is the
estimation of a graph's diameter, defined as the maximum weighted
distance between nodes in the same connected component.  For this
primitive, which is computationally intensive, resorting to
parallelism is inevitable as the graph size grows.  Unfortunately,
state of the art parallel strategies for diameter estimation are
either space inefficient or incur long critical paths.  These
strategies are thus unfeasible for dealing with huge graphs,
especially on distributed platforms characterized by limited local
memory and high communication costs, (e.g., clusters of
loosely-coupled commodity servers supporting a MapReduce-like
abstraction), which are widely used for big data tasks, and represent
the target computational scenario of this paper.  In this setting, the
challenge is to minimize the number of communication rounds while using
linear aggregate space  and small (i.e., substantially sublinear) space in the individual processors.

\vspace*{0.1cm}
\noindent
\emph{Previous work.}  For general graphs with arbitrary weights, the
only known approach for the exact diameter computation requires the
solution of the All-Pairs Shortest Paths (APSP) problem, but all known
APSP algorithms are characterized by space and/or time requirements
which make them impractical for very large graphs. For unweighted
graphs, the HyperANF algorithm by \cite{BoldiRV11} computes a tight
approximation to the eccentricity of each node (i.e., its maximum
distance from every other node), which in turn yields a tight
approximation to the graph diameter. HyperANF allows an
efficient multithreaded implementation and runs fast on a shared
memory platform. However, it has a critical path equal to the graph
diameter and requires a (small) non constant memory blow-up. 
Also, the algorithm cannot be adapted to deal with weighted graphs.
For these reasons, HyperANF is not a viable competitor 
for the objectives of the present work.

We observe that an upper bound to the diameter within a factor two
can be easily computed by solving an instance of the Single-Source
Shortest Path (SSSP) problem starting from an arbitrary node. A number
of recent works have explored the issue of reducing the approximation
factor by running several SSSP instances from carefully selected nodes
\cite{CrescenziGLM12,ChechikLRSTW14}.  While in practice a few such
instances suffice to obtain a very good approximation, worst-case
approximation guarantees less than two are obtained at the expense of
running a nonconstant number of SSSP instances. In fact, there is
strong theoretical evidence of the difficulty of approximating efficiently the
diameter within a factor less than 3/2 (see \cite{ChechikLRSTW14} and
references therein).

A number of parallel SSSP algorithms have been devised in the last two
decades \cite{Cohen00,klein_randomized_1997,meyer2003delta} (a more
extensive account of the literature can be found in
\cite{meyer2003delta}). To the best of our knowledge, the most
practical one is the \emph{$\Delta$-stepping} PRAM algorithm proposed
in \cite{meyer2003delta}.  The algorithm uses a design parameter
$\Delta$ to trade off parallel time and work. In a nutshell, the role
of $\Delta$ is to stagger the computation of the SSSP tree into zones
of bounded depth $\Delta$. Small values of $\Delta$ reduce work at the
expense of a longer critical path (i.e., parallel time) with the
algorithm approaching the behaviour of Dijkstra's algorithm; the
reverse tradeoff is obtained when $\Delta$ increases, approaching in
this case the behaviour of \mbox{Bellman-Ford's} algorithm
\cite{CormenLRS09}. It can be shown that, irrespective of $\Delta$, if
linear space is required, then the algorithm's parallel time is
bounded from below by the unweighted diameter of the input graph. In
fact, in \cite{meyer2003delta} the authors show that smaller parallel
time can be achieved by introducing suitable shortcut edges connecting
all node pairs at distance at most $\Delta$, which, however, have the
potential of resulting into superlinear space complexity, especially
for sparse graphs.

In \cite{CeccarelloPPU15} we devised a parallel $\BO{\log^3
  n}$-approximation algorithm for computing the diameter of unweighted
undirected graphs.  The algorithm first determines a partition of the
graph into $k = o(n)$ disjoint clusters through a decomposition
strategy based on growing the clusters from batches of centers
progressively selected from yet uncovered nodes, which yields a
provably small cluster radius. The diameter of the graph is then
obtained through the diameter of a suitable quotient graph associated
with the decomposition.  The algorithm can be efficiently implemented
in a MapReduce-like environment yielding good approximation quality in
practice. A similar clustering-based approach for diameter estimation
had been introduced in \cite{Meyer08} in the external-memory setting.
We remark that no analytical guarantees would be provided by the
weight-oblivious execution of these algorithms on a weighted graph
since, for a given topology, the system of shortest paths may
radically change once weights are introduced.

\vspace*{0.1cm}
\noindent
\emph{Our contribution.}  In this paper we address the more
challenging scenario of weighted graphs, and devise a diameter
approximation strategy for distributed machines supporting a
MapReduce-like abstraction, with provable theoretical guarantees,
providing extensive experimental evidence of its efficiency and
effectiveness.

Our algorithm adopts a cluster-based strategy similar to the one used
in \cite{CeccarelloPPU15} for the unweighted case.  
The main difficulty in the weighted case stems from the cluster growth
process: in the unweighted case we included all nodes in the frontier
of all active clusters at each round; in the weighted case we need to
avoid traversing heavy edges, which would increase the cluster radius
and have an unpredictable impact on the approximation quality.
%
%
On the other hand, a high rate of cluster growth is crucial
to enforce a small round complexity.  In order to tackle these
conflicting goals, we combine the following three main ingredients:
(1) the progressive cluster-growing strategy of
\cite{CeccarelloPPU15}; (2) an upper bound $\Delta$ on the weight of
the paths along which clusters are grown (an idea inspired by the
$\Delta$-stepping algorithm); and (3) a doubling strategy to guess
a quasi-optimal value of $\Delta$.

Analytically, we prove that with high probability our algorithm
attains an $\BO{\log^3 n}$ approximation ratio using a number of
rounds which is $\BO{\ell_R \log n}$, where $R$ is the maximum cluster
radius and $\ell_R$ is the number of edges required to connect any
pair of nodes at distance at most $R$. Both $R$ and $\ell_R$ are
nonincreasing functions of the number $k$ of clusters. To obtain the
best round complexity, $k$ can be chosen as the maximum value for
which the diameter of the quotient graph can be computed efficiently
in a single processor's local memory in one round. The total memory
space used by the algorithm is linear in the input graph size. 
We also show that on graphs of bounded doubling
dimension~\cite{AbrahamGGM2006} (an important family including, for
example, multidimensional arrays), under random edge weights and for
some positive constants $\delta<\epsilon<1$, if the local memory
available to each processor is $\BT{n^\epsilon}$ then the round
complexity of our algorithm can be made asymptotically smaller than
the unweighted diameter of the graph by a factor $\BT{n^\delta}$,
outperforming by at least that factor the approximation strategy based
on \dstepping, when linear space is required.
%

We complement the analytical results with an extensive set of
experiments on several large benchmark graphs (of up to about one
billion nodes and several billion edges), both real and synthetic,
running on a 16-node cluster where the Spark engine \cite{SPARK} is
employed to provide a MapReduce environment. For all graphs, our
algorithm attains an approximation ratio never exceeding 1.4, which is
far less than the theoretical upper bound and comparable with the one
ensured by the SSSP-based approach. Compared to an implementation of
$\Delta$-stepping in Spark, our algorithm is up to two orders of
magnitude faster. The better performance of our algorithm is also
substantiated by more implementation-independent metrics such
as the number of rounds and the aggregate work counted as number of
node updates and messages generated. We performed scalability tests
which, even in a somewhat small experimental testbed, demonstrate that
our algorithmic strategy has the potential to exploit a much higher
degree of parallelism, thus affording the analysis of much larger
graphs.

\vspace*{0.1cm}
\noindent
\emph{Paper organization.} Section~\ref{sec-prelim} introduces some
basic terminology and defines our reference machine
model. Section~\ref{sec-cluster} illustrates the graph decomposition
at the core of the diameter approximation algorithm, which is
described in the subsequent Section~\ref{sec-diameter}. The
experimental analysis is presented in Section~\ref{sec-experiments}.

\section{Preliminaries}\label{sec-prelim}
Let $G=(V,E,w)$ be a connected undirected weighted graph with $n$
nodes (set $V$), $m$ edges (set $E$), and a function $w$ which assigns
a positive integral weight $w(e)$ to each edge $e \in E$. We make the
reasonable assumption that the edge weights are polynomial in $n$.
(In fact, our results immediately extend to the case of positive
real-valued weights as long as the ratio between maximum and minimum
weight is polynomial in $n$.) The \emph{distance} between two nodes
$u, v \in V$, for short $\dist(u,v)$, is the weight of a
minimum-weight path between $u, v \in V$. The \emph{diameter} of the
graph, which we denote by $\Phi(G)$, is the maximum distance between
any two nodes.
%
%
We remark that while for convenience in the theoretical analysis we
consider only connected graphs, our results extend straightforwardly
to disconnected graphs, by defining the diameter as the largest
distance between any two nodes in the same connected component.

\begin{definition}
For any positive integer $\tau \leq n$, a \emph{$\tau$-clustering} of
$G$ is a partition $C = \{C_1, C_2 \dots, C_{\tau}\}$ of $V$ into
$\tau$ subsets called \emph{clusters}.  Each cluster $C_i$ has a
distinguished node \mbox{$u_i\in C_i$} called \emph{center}, and a
\emph{radius} $r(C_i) = \max_{v \in C_i} \{\dist(u_i,v)\}$.  The
\emph{radius of a $\tau$-clustering} $C$ is $r(C) = \max_{1 \le i \le
  \tau} \{r(C_i)\}$.  Finally, denote by $R_G(\tau)$ the minimum among
the radii of all $\tau$-clusterings of $G$.
\end{definition}

An important metric impacting the complexity of our algorithms is the
number of edges required to connect nodes with minimum-weight
paths. More formally, for a given distance parameter $\Delta$, we
define $\ell_{\Delta}$ to be the minimum value such that for any two
nodes $u,v \in V$ with $\dist(u,v) \leq \Delta$ there is a
minimum-weight path between $u$ and $v$ with at most $\ell_{\Delta}$
edges. It is easy to see that $\ell_{\Delta}$ is a nondecreasing
function of $\Delta$ and that for any constant $c>0$, $\ell_{c\Delta}
= \BT{\ell_{\Delta}}$.

The performance of the algorithms presented in this paper will be
analyzed considering their distributed implementation on the \emph{MR
  model} of \cite{PietracaprinaPRSU12}.  This model provides a
rigorous computational framework based on the popular MapReduce
paradigm \cite{DeanG08}, which is suitable for large-scale data
processing on clusters of loosely-coupled commodity servers. Similar
models have been recently proposed in \cite{KarloffSV10,GoodrichSZ11}.
An MR algorithm executes as a sequence of \emph{rounds} where, in a
round, a multiset $X$ of key-value pairs is transformed into a new
multiset $Y$ of pairs by applying a given reducer function (simply
called \emph{reducer}) independently to each subset of pairs of $X$
having the same key.  The model features two parameters $M_T$ and
$M_L$, where $M_T$ is the total memory available to the computation,
and $M_L$ is the maximum amount of memory locally available to each
reducer. We use \MR\ to denote a given instance of the model.  The
complexity of an \MR\ algorithm is defined as the number of rounds
executed in the worst case, and it is expressed as a function of the
input size and of $M_T$ and $M_L$. In the big data realm, practical
\MR\ algorithms should use total space linear in the input size and
significantly sublinear local space, while minimizing the round
complexity.

The following
fact is proven in \cite{GoodrichSZ11,PietracaprinaPRSU12}.
\begin{fact} \label{prefixsorting} \sloppy The sorting and (segmented)
  prefix-sum primitives for inputs of size $n$ can be performed in
  $\BO{\log_{M_L} n}$ rounds in \MR\ with $M_T=\BT{n}$.
\end{fact}

\section{Parallel graph decomposition}\label{sec-cluster}
In this section we present a parallel algorithm to partition a
connected undirected weighted graph $G=(V,E,w)$ into clusters of small
radius. The algorithm generalizes the one presented in
\cite{CeccarelloPPU15} for the unweighted case. Specifically, we grow
clusters in stages, where in each stage a new randomly selected batch
of cluster centers is added to the current clustering and the number
of uncovered nodes halves with respect to the preceding stage. 
The challenge in the weighted case is to perform cluster growth by
exploiting parallelism while, at the same time, limiting the weight of
the edges considered in each growing step: the goal is to avoid
increasing excessively the weighted radius of the clusters, which
influences the approximation quality.
%
%
In other words, unlike the strategy in
\cite{CeccarelloPPU15} for the unweighted case, we cannot afford to
boldly grow a cluster by adding all nodes connected to its frontier
since some of these additions may entail heavy edges.  To tackle this
challenge, we use ideas akin to those employed in the
$\Delta$-stepping parallel SSSP algorithm proposed in
\cite{meyer2003delta}. In particular, we limit a cluster's growth by
imposing a threshold $\Delta$ on its radius.  Unlike the
$\Delta$-stepping algorithm, however, this threshold is not fixed a
priori but automatically tuned to a quasi-optimal value during the
clustering process.

Before presenting the algorithm, we introduce some technical details.
Let $\Delta$ be an integral parameter which we use as a guess for the
radius of the clustering. As in \cite{meyer2003delta}, we call an edge
\emph{light} if its weight is $\leq \Delta$, and \emph{heavy}
otherwise.  With each node $u \in V$ the algorithm maintains a
\emph{state} consisting of two variables $(c_u,d_u)$: initially, $c_u$
is undefined and $d_u = \infty$. Whenever $u$ is assigned to a
cluster, then $c_u$ is set to the cluster center and $d_u$ is set to
an upper bound to $\dist(c_u,u)$. In particular, if $u$ is chosen as a
cluster center, then $c_u=u$ and $d_u=0$. The algorithm repeatedly
applies, in parallel, edge relaxations of the kind used in the
classical Bellman-Ford's algorithm. More precisely, we define the
following \emph{$\Delta$-growing step}: for each node $u$ with $d_u <
\Delta$ and for each light edge $(u,v)$, in parallel, if $d_u+w(u,v)
\leq \Delta$ and $d_v > d_u+w(u,v)$ then the status of $v$ is updated
by setting $d_v = d_u+w(u,v)$, and $c_v = c_u$.  In case more than one
node $u$ can provide an update to the status of $v$, the update that
yields the smallest value of $d_v$ and, secondarily, the one caused by
the node $u$ such that $c_u$ has smallest index, is performed.

Suppose that a sequence of $\Delta$-growing steps is performed
starting from some set of centers $X \subseteq V$.  After the
execution of these steps, we can contract the graph as follows
(Procedure Contract). For each center $c \in X$, all nodes $u \in V$
with $c_u =c$ are removed, except $c$ itself. For each edge $(u,v)$:
if both $c_u$ and $c_v$ are defined, the edge is removed; if both
$c_u$ and $c_v$ are undefined, the edge is left unchanged; and if
$c_u$ is defined and $c_v$ is undefined, the edge is replaced by a new
edge $(c_u,v)$ of weight $w(u,v)$.

Algorithm {\tt CLUSTER}$(G,\tau)$, whose pseudocode is given in
Algorithm~\ref{alg:cluster_growing}, grows clusters progressively in a
number of stages, until all the nodes of the graph are covered.  In
each stage, which corresponds to an iteration of the outer while loop,
a sequence of $\Delta$-growing steps is executed, each with the goal
of including into the current clusters at least half of the nodes that
are still uncovered but are reachable from some cluster through a path
of weight at most $\Delta$.  The set $X$ of centers of the current
clusters includes clusters partially grown in previous stages (if
any), which are now contracted and represented only by their centers
(set $C_i$), and a set of $\BO{\tau \log n}$ centers randomly selected
from the uncovered nodes.  In each stage, geometrically increasing
values of $\Delta$, starting from a suitable initial value, are
guessed until the coverage goal can be attained. (The sequences of
$\Delta$-growing steps for the various guesses of $\Delta$ are
performed in the inner while loop through Procedure {\tt
  PartialGrowth}.) When few nodes are left uncovered, these are added
as singleton clusters and the algorithm terminates.  Observe that at
most $\BO{\log n}$ stages are executed and each contributes an
additive factor $\Delta$ to the clustering radius. We will show below
that the largest guess for $\Delta$ will be $\BO{R_G(\tau)}$, with
high probability.

\begin{algorithm}[t]
\caption{{\tt CLUSTER}$(G,\tau)$}
\label{alg:cluster_growing}
\DontPrintSemicolon
\Let{$\Delta$}{$\min \{w(u,v) \; : \; (u,v) \in E\}$};
\Let{$\gamma$}{$4 \ln 2$}\;
\Let{$C_1$}{$\emptyset$} /*{\it current set of cluster centers} */\;
\Let{$G_1(V_1,E_1)$}{$G(V,E)$}\;
\Let{$i$}{1}\;

\While{$|V_i-C_i| \geq 8\tau \log n$}{

{
Select $v \in V_i-C_i$ as a new center
independently with probability
$\frac{(\gamma \tau \log n)}{|V_i-C_i|}$\;
}
\Let{X}{$C_i \cup \{\mbox{newly selected centers}\}$}\;
\ForEach{$u \in V_i$}{
\lIf{$u \in X$}{\Let{$(c_u,d_u)$}{$(u,0)$}
{\bf else} \Let{$(c_u,d_u)$}{$(\mbox{nil},\infty)$}}
}
\Let{$V'$}{$\emptyset$}\;
\While{$|V'| < |V_i-C_i|/2$}{
PartialGrowth$(G_i,\Delta)$ \;
\Let{$V'$}{$\{u \in V_i-C_i \; : \; d_u \leq \Delta$\}}\;
\lIf{$|V'| < |V_i-C_i|/2$}{\Let{$\Delta$}{$2\Delta$}}
}
Assign each $u \in V'$ to the cluster centered at $c_u$\;
\Let{$G_{i+1}(V_{i+1},E_{i+1})$}{Contract$(G_i)$}\;
\Let{$C_{i+1}$}{$X$}\;
\Let{$i$}{$i+1$}\;
}
Assign each $u \in V_i-C_i$ to a new singleton cluster centered at $u$\;
/* {\it $V_i$ is the final set of cluster centers} */\;
\vspace{7pt}
{\bf Procedure} PartialGrowth$(\bar{G},\Delta)$ \;
\Repeat{\rm (no state is updated) or ($|V'| \geq |V(\bar{G})|/2$)}
{
perform a $\Delta$-growing step on $\bar{G}$\;
\Let{$V'$}{$\{u \in V(\bar{G}) : d_u \leq \Delta\}$}\;
}
\end{algorithm}
Observe that when the algorithm terminates, each node $u \in V$ is
assigned to a cluster centered at some node $c_u$.  Let $\Dend$
denote the value of $\Delta$ at the end of the execution of
{\tt CLUSTER}$(G,\tau)$. In the following lemma, we show that with high
probability $\Dend$ does not exceed $R_G(\tau)$ by more than a 
constant factor.
\begin{lemma} \label{lem-clust2}
$\Dend = \BO{R_G(\tau)}$, with high probability.
\end{lemma}
\begin{proof}
Consider an arbitrary iteration of the outer while loop and let
$G_i=(V_i,E_i)$ be the (contracted) graph on which cluster growth is
performed during the iteration. Let $C_i \subseteq V_i$ be the nodes
of $G_i$ representing clusters grown in previous iterations.  Clearly,
we have that $|V_i-C_i| \geq 8\tau \log n$.  Refer to the nodes of
$V_i-C_i$ as \emph{uncovered nodes}.  We now show that, with
probability at least $1-1/n$, in $G_i$ at least half of the uncovered
nodes can be reached by the new centers selected in the iteration or
by nodes of $C_i$ with paths of weight at most $2R_G(\tau)$ traversing
only uncovered nodes.

Let $\bar{C}$ be a $\tau$-clustering of the whole graph with optimal
radius $r(\bar{C}) = R_G(\tau)$. To avoid confusion, we refer to its
clusters as \emph{$\bar{C}$-clusters}, while simply call
\emph{clusters} those grown by our algorithm.  Consider the
$\bar{C}$-clusters that include some uncovered node.  Among these
$\bar{C}$-clusters, those that contain less than $|V_i-C_i|/(2\tau)$
uncovered nodes account for a total of less than { $\tau{|V_i-C_i| / 2
    \tau} = {|V_i-C_i| / 2}$ } such nodes.  We call \emph{large} the
$\bar{C}$-clusters that contain $|V_i-C_i|/(2\tau)$ or more uncovered
nodes. Therefore, large $\bar{C}$-clusters account for more than
$|V_i-C_i|/2$ uncovered nodes altogether.  Let $\bar{C}_\ell$ be any
such large $\bar{C}$-cluster. By the choice of $\gamma$ in the
probability for center selection, we have that with probability 
$\geq (1-1/n^2)$ at least one uncovered node $c' \in \bar{C}_\ell$ is selected
as a new center.  Since, for every uncovered node $v \in
\bar{C}_\ell$, there is a path in $G$ from $c'$ to $v$ through
$\bar{c}_\ell$ of weight $w \leq 2 R_G(\tau)$, 
there must be a path in $G_i$ from from $c'$ to $v$ of weight at most
$w$ in $G_i$. Note that the suffix of this path starting from the last
cluster center has weight $\leq w$ and traverses only
unconvered nodes. The desired property follows by applying the 
union bound over all large $\bar{C}$-clusters. 

Since in PartialGrowth clusters are grown from the newly selected
centers as well as from the nodes of $C_i$, any value $\Delta \geq
2R_G(\tau)$ guarantees that half of the nodes in $V_i-C_i$ are covered
by clusters. Consequently, $\Delta$ can never be doubled beyond
$4R_G(\tau)$.  The lemma follows by applying the union bound over all
iterations.
\end{proof}
The following theorem states the main result of this section.
\begin{theorem}\label{thm:radius_batched}
\sloppy
Let $\tau$ be a positive integer. 
With high probability, {\tt CLUSTER}$(G,\tau)$ returns an $\BO{\tau
  \log^2 n}$-clustering of radius $\BO{R_G(\tau) \log n}$ by performing 
  $\BO{\ell_{R_G(\tau)}\log n }$ $\Delta$-growing steps, with
  $\Delta= \BO{R_G(\tau)} $.
  \end{theorem}
\begin{proof}
By Chernoff's bound each iteration of the outer {\bf while} loop
selects $\BO{\tau \log n}$ new cluster centers with high
probability. Hence, the bound on the number of clusters follows
applying the union bound over the $\BO{\log n}$ iterations of this
loop. As for the bound on the clustering radius, we observe that the
aggregate number of iterations of the inner {\bf while} loop is at most
$\log n + \log \Dend$, and this number is $\BO{\log n}$ with high
probability by virtue of Lemma~\ref{lem-clust2} and the assumption on
the polynomiality of the edge weights. The bound follows by observing
that each such iteration increases the radius of cluster by an
additive term at most $\Dend = \BO{R_G(\tau)}$.

Finally, observe that in every iteration of the
inner {\bf while} loop the number of $\Delta$-growing steps executed
by procedure PartialGrowth is at most $\ell_\Delta \leq \ell_{\Dend}$ since, by the properties of edge relaxations, after
those many growing steps all nodes at distance less than $\Delta$ from
some center of $X$ have been reached by the closest center in $X$ with
a minimum-weight path, hence their state cannot be further
updated. Therefore, the final number of $\Delta$-growing steps will be
$\BO{\ell_{\Dend}\log n}=\BO{\ell_{R_G(\tau)}\log n}$, with high probability.  
\end{proof}
\section{Diameter Approximation}\label{sec-diameter}
In this section, we present an algorithm to estimate the diameter of a
weighted graph as a function of the diameter of a suitable (much
smaller) quotient graph of a clustering obtained through a refined
version of the strategy devised in the previous section.  As it will
be clarified by the analysis, the refinement is introduced to achieve
a provable bound on the approximation guarantee, by ensuring
that not
too many clusters have the potential to reach small neighborhoods of
the graph.  Algorithm {\tt CLUSTER2}$(G,\tau)$, whose pseudocode is
given in Algorithm~\ref{alg:cluster2}, builds the required clustering
by first computing the radius $\ClustR$ of the clustering returned by
{\tt CLUSTER}$(G,\tau)$ and then executing $\log n$ iterations where
uncovered nodes are selected as new cluster centers with probability
doubling at each iteration.  In the $i$-th iteration, both previous
and new clusters are grown using $2\ClustR$-growing steps until
\emph{all} uncovered nodes at distance at most $2\ClustR$ from them are reached
(Procedure PartialGrowth2). At the end of the iteration, the graph is
contracted using Procedure Contract2, which is similar to Procedure
Contract used in {\tt CLUSTER} with the only difference that each
original edge $(u,v)$ of weight $w(u,v) \leq 2\ClustR$ and such that
$c_u$ is defined and $c_v$ is undefined, is replaced by a new edge
$(c_u,v)$ with rescaled weight $d_u+w(u,v)-2\ClustR$. (In fact,
original edges of weight greater than $2\ClustR$ are never used by
{\tt CLUSTER2}.)

\begin{algorithm}[t]
\caption{{\tt CLUSTER2}$(G,\tau)$}
\label{alg:cluster2}
\DontPrintSemicolon
Let $\ClustR$ be the radius of the clustering returned by {\tt CLUSTER}$(G,\tau)$\;
\Let{$C_1$}{$\emptyset$} /* {\it (current set of cluster centers)} */\;
\Let{$G_1(V_1,E_1)$}{$G(V,E)$}\;
\For{$i\leftarrow 1$ \KwTo $\log n$}{

{
Select $v\in V_i-C_i$ as a new center independently
with probability $2^{i}/n$\;
}

\Let{X}{$C_i \cup \{\mbox{newly selected centers}\}$}\;
\ForEach{$u \in V_i$}{
\lIf{$u \in X$}{\Let{$(c_u,d_u)$}{$(u,0)$}
  {\bf else} \Let{$(c_u,d_u)$}{$(\mbox{nil},\infty)$}}
}
PartialGrowth2$(G_i,2\ClustR)$ \;
\Let{$G_{i+1}(V_{i+1},E_{i+1})$}{Contract2$(G_i)$}\;
\Let{$C_{i+1}$}{$X$}\;
}
\vspace{7pt}
{\bf Procedure} PartialGrowth2$(\bar{G},\Delta)$ \;
{
\lRepeat{\rm no state is updated}
{perform a $\Delta$-growing step  on $\bar{G}$}
}
\end{algorithm}

\sloppy
The following lemma analyzes the quality of the clustering
returned by {\tt CLUSTER2}$(G,\tau)$ and upper bounds the number of
growing steps performed.
\begin{lemma}\label{lem-cluster2}
Let $\tau$ be a positive integer.  With high probability, {\tt
  CLUSTER2}$(G,\tau)$ computes an $\BO{\tau \log^4 n}$-clustering of
radius $\Clust2R = \BO{R_G(\tau) \log^2 n}$ by performing
$\BO{\ell_{R_G(\tau) \log n}\log n}$ $\Delta$-growing steps with
$\Delta = \BO{R_G(\tau) \log n}$.
\end{lemma}
\begin{proof}
By Theorem~\ref{thm:radius_batched} we know that with high probability
the invocation of {\tt CLUSTER}$(G,\tau)$ at the beginning of the
algorithm computes a $K$-clustering, with $K = \BO{\tau \log^2n}$, of
radius $\ClustR = \BO{R_G(\tau) \log n}$.  In what follows, we
condition on this event. Then, the bounds on the number of growing
steps and on $\Clust2R$ are straightforward.  For $\gamma = 4/\log_2
e$,  define $H$ as the smallest integer such that $2^H/n \geq
(\gamma K \log n)/n$, and let $t=\log n -H$. We now show that the
number of original nodes of $G$ not yet reached by any cluster
decreases at least geometrically at each iteration of the {\bf for}
loop after the $H$-th one. Recalling that $V_{H+i}-C_{H+i}$ is the set
of original nodes of $G$ that at the beginning of Iteration $H+i$ have
not been reached by any cluster, for $1\leq i\leq t$, define the event
$E_i =$``at the beginning of Iteration $H+i$, $V_{H+i}-C_{H+i}$
contains at most $n/2^{i-1}$ nodes''.  We now prove that the event
$\cap_{i=1}^{t} E_i$ occurs with high probability. Observe that:
\begin{eqnarray*}
\Pr\left(\cap_{i=1}^{t} E_i\right) 
& = & \Pr(E_1)\prod_{i=1}^{t-1} \Pr(E_{i+1} | E_1\cap \cdots \cap E_i) \\
& = & \prod_{i=1}^{t-1} \Pr(E_{i+1} | E_1\cap \cdots \cap E_i),
\end{eqnarray*}
since $E_1$ clearly holds with probability one.  Consider an arbitrary
$i$, with $1 \leq i < t$, and assume that $E_1\cap \cdots \cap
E_i$ holds. We prove that $E_{i+1}$ holds with high probability.
Since $E_i$ holds, we have that at the beginning of
Iteration $H+i$, the number of nodes in $V_{H+i}-C_{H+i}$
is at most $n/2^{i-1}$. Clearly, if $|V_{H+i}-C_{H+i}| \leq n/2^i$ then $E_{i+1}$
trivially holds with probability one. Thus, we consider only the case
\[
{n \over 2^i} < |V_{H+i}-C_{H+i}| \leq {n \over 2^{i-1}}.
\]
In order to show that $E_{i+1}$ holds also in this case, we resort to the 
same argument used in the proof of Lemma~\ref{lem-clust2}.
Let $\bar{C}$ be a $K$-clustering of the whole graph with optimal
radius $R_G(K)$ and observe that $R_G(K) \leq \ClustR$.  To avoid
confusion, we refer to its clusters as \emph{$\bar{C}$-clusters},
while simply call \emph{clusters} those grown by {\tt CLUSTER2}.
Consider the $\bar{C}$-clusters that include some nodes of
$V_{H+i}-C_{H+i}$, and call one such cluster \emph{large} if it
contains at least 
\[
{|V_{H+i}-C_{H+i}| \over 2K} > {n \gamma \log n \over 2^{H+i+1}}
\]
nodes of $V_{H+i}-C_{H+i}$. This implies that the large
$\bar{C}$-clusters contain, altogether, at least half of the nodes of
$V_{H+i}-C_{H+i}$. Moreover, by the choice of $\gamma$, it is easy to
argue that with probability $\geq (1-1/n)$ at least one new center is
selected  from each large $\bar{C}$-cluster in Iteration $H+i$.  Consider
now an arbitrary large $\bar{C}$-cluster centered at $\bar{c}_\ell$
and let $c' \in \bar{C}_\ell$ be a new center selected from this
cluster in the iteration.  For every $v \in \bar{C}_\ell \cap
(V_{H+i}-C_{H+i})$ there is a path in $G$ from $c'$ to $v$ (through
$\bar{c}_\ell$) of weight $w \leq \ClustR$, hence there must be a path
in $G_i$ from from $c'$ to $v$ of weight at most $w$.  It then follows that node $v$
will be covered by some cluster in Iteration~$H+i$.  Consequently, in the
iteration at least half of the nodes of $V_{H+i}-C_{H+i}$ will be
covered by clusters, with probability at least $1-1/n$.

By multiplying the probabilities of the $\BO{\log n}$ conditioned
events, we conclude that event $\cap_{i=1}^{t} E_i$ occurs with high
probability. Note that in the last iteration (Iteration~$H+t$) all
uncovered nodes are selected as centers with probability 1, and, if
$\cap_{i=1}^{t} E_i$ occurs, these are $\BO{K \log n}$. Now, one can
easily show that, with high probability, in the first $H$ iterations,
$\BO{K \log^2 n}$ clusters are added and, by conditioning on
$\cap_{i=1}^{t+1} E_i$, at the beginning of each Iteration~$H+i$,
$1\leq i\leq t$, $\BO{K\log n}$ new clusters are created, for a total
of $\BO{K\log^2n} = \BO{\tau\log^4n}$ clusters.
\end{proof}

Observe that for fixed $\tau$, the clustering returned by {\tt CLUSTER2} 
has a larger number of clusters and a weaker guarantee on its radius than the
the clustering returned by {\tt CLUSTER}. As such, {\tt CLUSTER2}  does not
appear to be a very desirable clustering strategy in itself. However, {\tt CLUSTER2}
enforces the following  important property which will be needed for proving  the diameter approximation.
With reference to a specific execution of  {\tt CLUSTER2}, define the \emph{light distance} between two nodes $u$ and $v$ as the weight
of the minimum-weight path from $u$ and $v$ consisting only of edges
of weight at most $2\ClustR$. (Note that the light distance is not
necessarily defined for every pair of nodes.) 

Due to the weight rescaling performed by Contract2 at the end of each
iteration, given a center $c$ selected at a certain Iteration $i$ of
the {\bf for} loop, and a node $v$ at light distance $d$ from $c$, the
cluster centered at $c$ cannot grow to reach $v$ in less than
$\lceil d/2\ClustR\rceil$ iterations and that in those many iterations
$v$ will be reached by some cluster (possibly the one centered at
$c$). Consequently, no center selected at a later iteration
at that same distance from $v$ as $c$ would be able to reach $v$.

We are now ready to present the main result of this section, which
shows how {\tt CLUSTER2} can be employed to determine a good
approximation to the graph diameter. Suppose we run {\tt CLUSTER2} on
a graph $G=(V,E,w)$ to obtain a clustering $C$ of radius $\Clust2R$.
For each $u \in V$, let $c_u$ be the center of the cluster assigned to
$u$, and let $d_u$ be distance between $u$ and $c_u$ returned by {\tt
  CLUSTER2}.  As in \cite{Meyer08}, we define the weighted quotient
graph associated to $C$ as the graph $G_C$ where nodes correspond to
clusters and, for each edge $(u,v)$ of $G$ with $c_u \neq c_v$, there
is an edge in $G_C$ between the clusters of $u$ and $v$ with weight
$w(u,v)+d_u+d_v$. (In case of multiple edges between two clusters, it
is sufficient to retain only the one yielding minimum weight.)  Let
$\Phi(G)$ (resp., $\Phi(G_C)$) be the weighted diameter of $G$ (resp.,
$G_C$). We approximate $\Phi(G)$ through the value $\Phi_{\rm
  approx}(G) = \Phi(G_C)+2 \Clust2R$.  It is easy to see that our
estimate is conservative, that is, $\Phi_{\rm approx}(G) \geq
\Phi(G)$. We have:
\begin{theorem}\label{segment}
With high probability,
{
\[
\Phi_{\rm approx}(G) = \BO{\Phi(G) \log^3 n}.
\]
}
\end{theorem}
\begin{proof}
Since $\Clust2R = \BO{R_G(\tau) \log^2n}$ (by
Lemma~\ref{lem-cluster2}), and $R_G(\tau) = \BO{\Phi(G)}$, we have
that $\Clust2R = \BO{\Phi(G) \log^2n}$.  In order to show that
$\Phi(G_C) = \BO{\Phi(G) \log^3 n}$, let us fix an arbitrary pair of
cluster centers and an arbitrary minimum-weight path $\pi$ between
them in $G$, and let $w_\pi$ be the weight of $\pi$.  Let $\pi_C$ be
the path of clusters in $G_C$ traversed by $\pi$.  We now show that
with high probability the weight of $\pi_C$ in $G_C$ is $\BO{\Phi(G)
  \log^3 n}$. The bound on the approximation will then follow by
applying the union bound over all pairs of cluster centers.  Consider
first the case $2\ClustR > \Phi(G)$ (note that this can happen since
the clustering yielding the radius $\ClustR$ determined at the
beginning of {\tt CLUSTER2} is built out of paths using only light
edges of weight $\BO{R_G(\tau)}$). In this case, it is easy to see
that the first batch of centers ever selected in an iteration of the
for loop of {\tt CLUSTER2} will cover the entire graph, and these
centers are $\BO{\log n}$ with high probability. Therefore, $\pi_C$
contains $\BO{\log n}$ clusters and its weight is
$\BO{w_{\pi}+\Clust2R \log n}=\BO{\Phi(G) \log^3 n}$. Suppose now that
$2\ClustR \leq \Phi(G)$.  We show that at most $\BO{\lceil
  w_\pi/\ClustR \rceil\log^2 n}$ clusters intersect $\pi$ (i.e.,
contain nodes of $\pi$), with high probability.  It can be seen that
$\pi$ can be divided into $\BO{\lceil w_{\pi} /\ClustR\rceil}$
subpaths, where each subpath is either an edge of weight $>
\ClustR$ or a segment of weight $\leq \ClustR$. It is then sufficient to
show that the nodes of each of the latter segments belong to
$\BO{\log^2 n}$ clusters. Consider one such segment $S$. Clearly, all clusters
containing nodes of $S$ must have their centers at light distance at
most $\Clust2R$ from $S$ (i.e., light distance at most $\Clust2R$ from
the closest node of $S$).  Recall that $\Clust2R \leq 2\ClustR \log
n$. For $1 \leq j \leq 2\log n+1$, let $C(S,j)$ be the set of nodes
whose light distance from $S$ is between $(j-1)\ClustR$ and
$j\ClustR-1$, and observe that any cluster intersecting $S$ must be
centered at a node belonging to one of the $C(S,j)$'s. We claim that,
with high probability, for any $j$, there are $\BO{\log n}$ clusters
centered at nodes of $C(S,j)$ which may intersect $S$.  Fix an index
$j$, with $1 \leq j \leq 2\log n+1$, and let $i_j$ be the first
iteration of the for loop of {\tt CLUSTER2} in which some center is
selected from $C(S,j)$. 
By the property of
{\tt CLUSTER2} discussed after Lemma~\ref{lem-cluster2},
$\lceil ((j+1)\ClustR-1)/(2\ClustR) \rceil$
iterations are sufficient
for any of these centers to cover the entire segment. On the other hand,
any center from $C(S,j)$ needs at least $\lceil (j-1)\ClustR/(2\ClustR) \rceil$ iterations to touch the segment. Hence, 
 we have that no center selected
from $C(S,j)$ at Iteration $i_j+2$ or higher is able to reach $S$.
It is easy to see that, due to the smooth
growth of the center selection probabilities, the number of centers
selected from $C(S,j)$ in Iterations $i_j$ and $i_j+1$
is $\BO{\log n}$, with high
probability. 
This implies that the nodes of segment $S$ will belong to
$\BO{\log^2 n}$ clusters, with high probability. By applying the union
bound over all segments of $\pi$, we have that $\BO{\lceil
  w_\pi/\ClustR \rceil\log^2 n}$ clusters intersect $\pi$, with high
probability. Therefore, the weight of $\pi_C$ is $\BO{w_{\pi}+\Clust2R
  \lceil w_\pi/\ClustR \rceil\log^2 n} = \BO{\Phi(G)+\Clust2R
  (\Phi(G)/\ClustR) \log^2 n} = \BO{\Phi(G) \log^3 n}$.  
\end{proof}
\subsection{Implementation in the \MR\ model}
We now discuss the implementation of the above diameter approximation
algorithm in the MR model using overall linear space and show that,
for a relevant class of graphs, its round complexity can be made
asymptotically smaller than the one required to obtain a
2-approximation through the state-of-the-art SSSP algorithm by
\cite{meyer2003delta}. For a given connected weighted graph
$G=(V,E,w)$, consider the \MR\ model with total memory $M_T$ linear in
the graph size, and local memory $M_L \in \BO{n^\epsilon}$, for some
constant $\epsilon \in (0,1)$. We begin by observing that, regardless
of the number of active clusters, a $\Delta$-growing step, for any
$\Delta$, can be implemented through a constant number of simple
prefix and sorting operations which, by Fact~\ref{prefixsorting},
require $\BO{1}$ rounds on \MR. By combining this observation with
the results of Theorem~\ref{thm:radius_batched} and
Lemma~\ref{lem-cluster2}, we obtain that {\tt CLUSTER2}$(G,\tau)$ can be
implemented in $\BO{\ell_{R_G(\tau) \log n}\log n}$ rounds in \MR.

For an arbitrary positive constant $\epsilon'< \epsilon$, 
let $\tau = \lceil n^{\epsilon'} \rceil$
and let $G_C=(V_C,E_C,w_C)$ be the weighted quotient graph
associated with the clustering returned by {\tt CLUSTER2}$(G,\tau)$.  The
diameter $\Phi(G_C)$ (in fact, a constant approximation to this
quantity), and, consequently, the value $\Phi_{\rm approx}(G)$, can
then be computed in $\BO{1}$ rounds in \MR\ by adopting the same
techniques described in \cite{CeccarelloPPU15}.

The following theorem summarizes the above discussion.
\begin{theorem} \label{MR-complexity}
Let $G$ be a connected weighted graph with $n$ nodes, $m$ edges and
weighted diameter $\Phi(G)$. Also, let $\epsilon' < \epsilon \in
(0,1)$ be two arbitrary constants,
and let $\tau = \lceil n^{\epsilon'} \rceil$. 
On the \MR\ model, with $M_G =
\BT{m}$ and $M_L = \BT{n^{\epsilon}}$, an upper bound $\Phi_{\rm
  approx}(G) = \BO{\Phi(G) \log^3 n}$ to the diameter of $G$ can be
computed in $\BO{\ell_{R_G(\tau) \log n}\log n}$ rounds, with high 
probability.
\end{theorem}
Note that the round complexity depends on the characteristics of the
graph and is nonincreasing in the number of clusters, which are in
turn controlled by parameter $\tau$. For general graphs, the value
$\ell_{R_G(\tau) \log n}\log n$ is $\BO{\ell_{\Phi(G)}\log n}$ while
the analysis in \cite{meyer2003delta} implies that under the
linear-space constraint a natural MR-implementation of $\Delta$-stepping
requires $\BOM{\ell_{\Phi(G)}\log n}$ rounds. Hence, our algorithm
cannot be asymptotically slower than $\Delta$-stepping. However, we
show below that our algorithm becomes considerably faster for an
important class of graphs.

The following definition introduces a concept that a number of recent works
have shown to be useful in relating algorithms' performance to graph
properties~\cite{AbrahamGGM2006}.
\begin{definition} \label{doublingdim} Consider an
undirected graph $G=(V,E)$.  The \emph{ball of radius $R$} centered at
node $v$ is the set of nodes reachable through paths of at most $R$
edges from $v$. Also, the \emph{doubling dimension} of $G$ is the
smallest integer $b > 0$ such that for any $R >0$, any ball of radius
$2R$ can be covered by at most $2^b$ balls of radius $R$.
\end{definition}
We can specialize the result of Theorem~\ref{MR-complexity} as follows.
\begin{corollary} \label{MR-complexity2}
Let $G$ be a connected graph with $n$ nodes, $m$ edges, maximum degree
$d \in \BO{1}$, doubling dimension $b \in \BO{1}$, and 
positive integral edge weights chosen uniformly at random from a
polynomial range. Denote by $\Phi(G)$ and
$\Psi(G)$, respectively, the weighted and unweighted diameter of
$G$. Also, let $\epsilon' < \epsilon \in (0,1)$ be two arbitrary
constants.  On the
\MR\ model, with $M_G = \BT{m}$ and $M_L = \BT{n^{\epsilon}}$, an
upper bound $\Phi_{\rm approx}(G) = \BO{\Phi(G) \log^3 n}$ to the
diameter of $G$ can be computed in 
\[
\BO{\left\lceil \Psi(G) \over n^{\epsilon'/b} \right\rceil \log^3 n}
\]
rounds, with high probability.
\end{corollary}

\begin{proof}
By iterating the definition of doubling dimension starting from a
single ball of unweighted radius $\Psi(G)$ containing the whole graph,
we can decompose the graph into $\tau$ disjoint clusters of unweighted
radius $\psi=\BO{\ceil{\Psi(G)/\tau^{1/b}}}$.  Letting $W$ be the
maximum edge weight, we have that $\psi \cdot W$ upper bounds
$R_G(\tau)$.  We know that our algorithm computes the diameter
approximation in $\BO{\ell_{R_G(\tau)\log n}\log n}$ rounds. We will
now give an upper bound on $\ell_{R_G(\tau)\log n}$. By using results
from the theory of branching processes \cite{Dwass69} we can prove
that by removing all edges of weight $\geq \lfloor W/d \rfloor$ with
high probability the graph becomes disconnected and each connected
component has $\BO{\log n}$ nodes. (More details will be provided in
the full version of the paper.)  As a consequence, with high
probability any simple path in $G$ will traverse an edge of weight
$\geq \lfloor W/d \rfloor = \BOM{W}$ every $\BO{\log n}$ nodes. This implies
that a path of weight at most $R_G(\tau)\log n =
\BO{\ceil{\frac{\Psi}{\tau^{1/b}}}W\log n}$ has $\ell_{R_G(\tau)\log
  n} = \BO{\ceil{\frac{\Psi}{\tau^{1/b}}}\log^2 n}$ edges.  The
theorem follows by setting $\tau = \lceil n^{\epsilon'} \rceil$.
\end{proof}
For what concerns the comparison with $\Delta$-stepping, the analysis
in \cite{meyer2003delta} implies that under the linear-space
constraint, for a graph $G$ with random uniform weights a natural
MR-implementation of
$\Delta$-stepping requires $\BOM{\Psi(G)}$ rounds. Thus, by the above
corollary, if $G$ has bounded doubling dimension the round complexity
of our algorithm can be made smaller by a sublinear yet polynomial
factor which is a function of the available local space $M_L$.

We conclude this section by observing that in the case of very skewed
graph topologies and/or weight distributions under which the
hypotheses of Corollary~\ref{MR-complexity2} do not hold, the factor
$\ell_{R_G(\tau)\log n}$ in the round complexity of our algorithm can
be large, thus reducing the competitive advantage with respect to
\dstepping. We can somewhat overcome this limitation by imposing an
upper limit $\BO{n/\tau}$ (resp., $\BO{(n/\tau)\log n}$) to the number
of growing steps performed in each execution of PartialGrowth within
{\tt CLUSTER}$(\tau)$ (resp., PartialGrowth2 within {\tt
  CLUSTER2}$(\tau)$). It can be shown that, in this case, the round
complexity of our algorithm, for general graphs, becomes
$\BO{\min\{(n/\tau)\log n,\ell_{R_G(\tau) \log n}\}\log n}$ at the
expenses of an extra
$\BO{\lceil \ell_{R_G(\tau) \log n} /((n/\tau)\log n) \rceil}$ factor
in the approximation ratio. The argument revolves around the existence
of quasi-optimal clusterings of bounded unweighted depth. (More
details will be provided in the full version of this paper.)

\section{Experimental Analysis}
\label{sec-experiments}

\begin{table}
  \centering
  \begin{tabular}{l@{\hskip 1pt} r r r}
    \toprule
    Graph & $n$ & $m$ & $\Phi(G)$ \\
    \midrule
    {\tt roads-USA}~\cite{DIMACS} 
          & 23,947,347 & 29,166,673 & 55,859,820 \\
    {\tt roads-CAL}~\cite{DIMACS} 
          & 1,890,815 & 2,328,872 & 16,485,258 \\
    {\tt livejournal}$\star$~\cite{SNAP} 
          & 3,997,962 & 32, 681, 189 & 9.41 \\
    {\tt twitter}$\star$~\cite{LAW} 
          & 41,652,230 & 1,468,365,182 & 9.07 \\
    {\tt mesh(S)}$\star$ 
          & $S^2$ & $2S(S-1)$ & $\dagger$ \\
    {\tt R-MAT(S)}$\star$~\cite{Chakrabarti2004} 
          & $2^S$ & $16\cdot 2^S$ & $\dagger$ \\
    {\tt roads(S)} 
          & $\approx S\cdot 2.3\cdot 10^7$ & $\approx S\cdot 5.3\cdot 10^7$ & $\dagger$ \\
    \bottomrule
    \multicolumn{4}{l}{
    $\dagger$ the diameter depends on the size of
    the graph, controlled by $S>1$.
    }
  \end{tabular}
  \caption{Benchmark graphs: $n$ is the number of nodes, $m$ the
    number of edges and $\Phi(G)$ is the weighted diameter. Graphs marked
    with $\star$ have edge 
    weights following a random uniform distribution $\in (0, 1]$. The
    {\tt twitter} graph, originally directed, has been symmetrized.}
  \label{tab:benchmark-graphs}
\end{table}

Our experimental platform is a cluster of 16 nodes, each equipped with
a 4-core I7 processor and 18GB RAM, connected by a 10Gbit Ethernet
network. Our algorithms are implemented using Apache
Spark~\cite{SPARK}, a popular framework for big data computations that
supports the MapReduce abstraction adopted by our algorithms. The
experiments have been run on several graphs whose properties are
summarized in Table~\ref{tab:benchmark-graphs} and can be classified
as follows: a) road networks ({\tt roads-USA} and {\tt roads-CAL}), b)
social networks ({\tt livejournal} and {\tt twitter}), c) synthetic
graphs ({\tt mesh(S)}, {\tt R-MAT(S)}, and {\tt roads(S)}, where {\tt
  S} is a parameter controlling the size of the graph). The latter
class contains artificially generated graphs whose size can be made
arbitrarily large and whose topological properties reflect those of
the real networks in the first two classes. In particular, {\tt
  R-MAT(S)} are graphs with a power-law degree distribution and small
diameter~\cite{Chakrabarti2004}, and {\tt roads(S)} are graphs
obtained as the cartesian product of a linear array of $S$ nodes $S
>1$ and unit edge weights with {\tt roads-USA}. Finally, {\tt mesh(S)}
is an $S \times S$ square mesh included since it is a graph of known
doubling dimension $b=2$ for which the results of
Corollary~\ref{MR-complexity2} hold.  All road networks come with
original integer weights while for the other graphs, which are born
unweighted, we assigned uniform random edge weights in $(0,1]$ according
  to the approach commonly adopted in the literature.

We implemented a simplified version of our diameter approximation
algorithm, dubbed \cldiam, where, for efficiency, we used {\tt
  CLUSTER}, rather than {\tt CLUSTER2}, for computing the graph
decomposition. In fact, {\tt CLUSTER2} first runs {\tt CLUSTER} to
obtain an estimate of the radius, and then computes a second
decomposition which is instrumental to provide a theoretical bound to
the approximation factor, but which does not seem to provide a
significant improvement to the quality of the approximation in
practice.

As a second optimization, we ran {\tt CLUSTER} using an initial value
of $\Delta$ larger than the minimum edge weight, as was specified in
the pseudocode. We observe that by increasing the initial value of
$\Delta$, the round complexity improves since less doublings are
required before hitting the final value. On the other hand, setting
the initial value of $\Delta$ too large may yield a larger cluster
radius, possibly incurring a worse diameter approximation. To explore this
phenomenon, we experimented on {\tt mesh}$(S)$ with $S=2048$
and random
edge weights, such that an edge has weight $1$ with probability $0.1$
and $10^{-6}$ otherwise. With high probability, such a graph can be
completely covered using clusters that do not contain edges with
weight 1: including one of those edges in a cluster would make its
radius far bigger than it needs to be. We ran our algorithm with two
configurations. The first configuration started with $\Delta =
10^{-6}$ (i.e., the minimum edge weight) so to let the algorithm
tune itself to the final value $\Delta$ ($=6.4\cdot 10^{-5}$);
the second configuration started with an inital $\Delta$ equal to the
graph diameter ($\approx 2.004$) so that no doubling of $\Delta$ was
needed. The diameter approximation obtained by the second
configuration was about $2.5$ times larger than the actual diameter,
whereas the first configuration obtained an approximation ratio of
$1.0001$. A set of experiments (omitted here for brevity) showed that
a good initial guess for $\Delta$ is the average edge weight, which
reduces the round complexity without affecting the approximation
quality significantly. Therefore, all our experiments have been run
with this initial guess of $\Delta$.

Finally, in all of our experiments the parameter $\tau$ was set to
yield a number of nodes in the quotient graph $\le 100,000$. This
choice of $\tau$ was made to ensure that, for all instances, the final
diameter computation in the quotient graph would not dominate the
running time. 

The following paragraphs describe in detail the different sets of
experiments that we performed.

\vspace*{0.1cm}
\noindent
\emph{Comparison with the SSSP-based approximation.}  
Recall that an SSSP algorithm
 can be used to yield a 2-approximation to the diameter 
 by returning twice the weight of the heaviest shortest path.
Thus, we compared our
algorithm \cldiam with a Spark implementation
of the \dstepping SSSP algorithm (starting from a random node), 
which is the state of the art for parallel SSSP
and is in fact our only practical competitor on weighted graphs. 
In \dstepping, parameter $\Delta$ can be set to control the tradeoff
between parallel time (i.e., rounds in the MapReduce context) and total
work. For each graph, we tested  \dstepping
with several values of $\Delta$, selecting the value 
yielding the best running time. Since in MapReduce-like  
environments
the number of rounds
has a significant impact on the running time, not surprisingly, for all graphs the best value
of $\Delta$ was always the one minimizing the number of rounds.

The results of the comparison are summarized in
Table~\ref{tab:comparison-dstepping} and graphically represented in
Figures~\ref{fig:approximation}, \ref{fig:num-rounds},
and~\ref{fig:work}. In the table we report, for each graph, the diameter
approximation factor and the running time. Along with these, we also
report two additional measures, namely, the number of rounds and the
work (defined as the sum of node updates and messages generated), that
allow to compare the two algorithms in a more platform-independent
way. It has to be remarked that the approximation quality returned by
\cldiam\ on all benchmark graphs, a value always less that $1.4$, is
much better that the theoretical $\BO{\log^3n}$ bound.
Also, our algorithm is from about one to two orders of magnitude
faster than \dstepping, while featuring comparable approximation
ratios (Figure~\ref{fig:approximation}). As expected, the higher
performance of \cldiam is consistent with the fact that it requires
far less rounds than \dstepping, as shown in
Figure~\ref{fig:num-rounds}.

\newcommand{\samount}{10pt}
\begin{table*}[t]
  \centering\footnotesize
  \begin{tabular}{l@{\hskip 1pt} 
                  r@{\hskip \samount} r@{\hskip \samount} 
                  r@{\hskip \samount} r@{\hskip \samount} 
                  r@{\hskip \samount} r@{\hskip \samount} 
                  r@{\hskip \samount} r}
    \toprule
    & \multicolumn{2}{c}{approximation}
    & \multicolumn{2}{c}{time}
    & \multicolumn{2}{c}{rounds}
    & \multicolumn{2}{c}{work} \\
    graph
    & \cldiam & \dstepping & \cldiam & \dstepping
    & \cldiam & \dstepping & \cldiam & \dstepping \\
    \midrule

{\tt roads-USA  } &          1.26 &              1.09 &     158 &            14,982 &      74 &            11,268 & 4.22\ord{8} &          1.35\ord{11} \\
{\tt roads-CAL  } &          1.25 &              1.22 &      13 &          917 &      16 &             2,639 & 1.70\ord{7} &          2.27\ord{9}~ \\
{\tt mesh       } &          1.23 &              1.30 &      46 &             1,239 &      70 &             2,997 & 1.13\ord{8} &          1.58\ord{8}~ \\
{\tt livejournal} &          1.22 &              1.10 &      19 &               
 74 &       7 &                42 & 1.97\ord{8} &          4.81\ord{8}~ \\
{\tt twitter    } &          1.19 &              1.32 &     236 &               
601 &       5 &                35 & 4.71\ord{9} &          1.20\ord{10} \\
{\tt R-MAT(24)  } &          1.33 &              1.02 &     144 &             1,493 &       4 &                41 & 7.45\ord{8} &          4.20\ord{9}~ \\

    \bottomrule
  \end{tabular}
  \caption{\cldiam vs \dstepping.
    For each benchmark, the table shows the running time (in
    seconds), the approximation ratio, the number of rounds, and the
    work of the two algorithms.
    The approximation ratio is expressed in terms of a lower bound to
    the true diameter computed by running the sequential SSSP
    algorithm multiple times, each time starting from the farthest
    node reached by the previous run.
  }
    \label{tab:comparison-dstepping}
\end{table*}

\if \onepagefig0
\begin{figure}
  \centering
  \includegraphics[width=.45\columnwidth]{plots/approximation}
  \caption{Approximation ratio of \cldiam and
    \dstepping.
    \vspace{40pt}}
  \label{fig:approximation}
\end{figure}
\else
\fi

\if \onepagefig0
\begin{figure}
  \centering
  \includegraphics[width=.5\columnwidth]{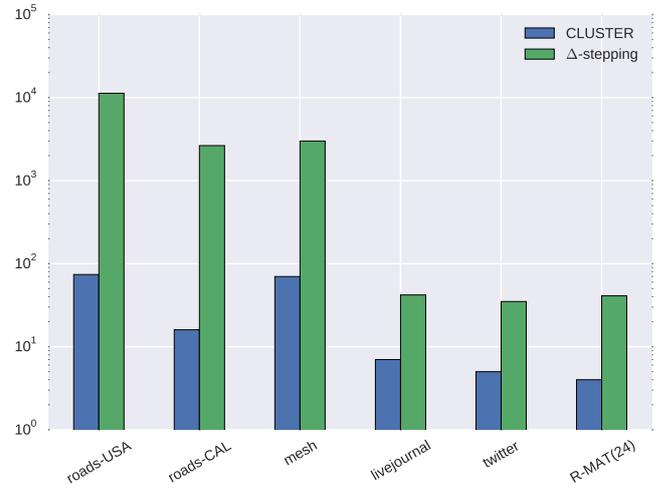}
  \caption{Number of rounds required by \cldiam and
    \dstepping. The scale is logarithmic.}
  \label{fig:num-rounds}
\end{figure}
\else
\fi

For what concerns the work, the better performance of \cldiam, as
shown in Figure~\ref{fig:work}, is mainly due to the smaller number
relaxations with respect to \dstepping. Indeed, \cldiam explores paths
only up to a limited depth, whereas \dstepping needs to run until all
the nodes are labeled with the optimal distance from the source. In
fact, \dstepping could limit the amount of relaxations by using a
smaller $\Delta$, but in doing so it would incur an increase of the
number of rounds, hence exhibiting worse performance.
The gap between \dstepping and
\cldiam in terms of both number of rounds and work 
suggests that our algorithm is likely to remain
competitive on other distributed-memory platforms employing 
programming frameworks alternative to MapReduce.

\if \onepagefig0
\begin{figure}
  \centering
  \includegraphics[width=.5\columnwidth]{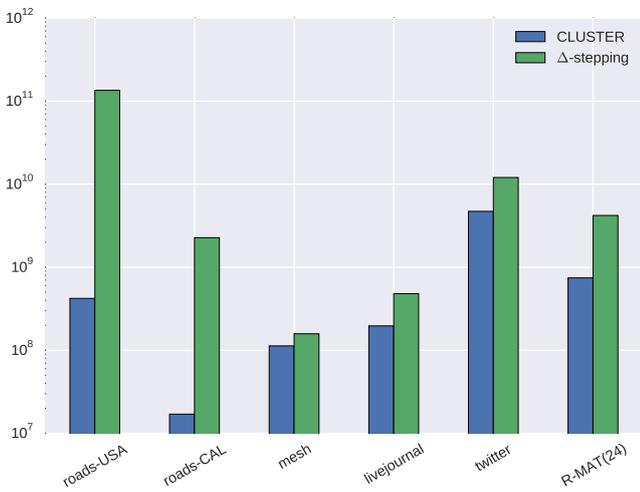}
  \caption{Work performed by \cldiam and \dstepping. The scale
    is logarithmic.}
  \label{fig:work}
\end{figure}
\else
\fi

We remark that the experiments reported in
Table~\ref{tab:comparison-dstepping} involve graphs of moderate size
for which, not surprisingly, much better running times can be obtained
on a single machine equipped with sufficient main memory.  In fact,
the purpose of those experiments was not to attain best absolute
performance on the individual graphs but, rather, to compare the
relative performance of \cldiam with the one of \dstepping.
%
Since it is conceivable to expect that the relative performance of two
algorithms does not change as the graphs size grows, performing this
comparison on much larger graphs would have only encumbered the
experimental work without changing the overall outcome. Nevertheless,
we performed further experiments, reported in the next subsection, to
provide evidence that our algorithm scales well with respect to
machine and input size.

\vspace*{0.1cm}
\noindent
\emph{Scalability.}
To check the scalability of \cldiam with respect to the number of
machines, we ran it using 2, 4, 8, and 16 machines on {\tt R-MAT(26)}
and {\tt roads(3)} which have approximately the same number of nodes
but have topologies of different nature. The results are reported in Figure
\ref{fig:scalability} which shows that, for both graphs, 
the algorithm exhibits excellent scalability. 

Finally, we ran \cldiam on {\tt R-MAT(29)} and {\tt roads(32)}, which
are much larger graphs than those employed in the experiments reported
in Table~\ref{tab:comparison-dstepping}, and for which the running
time of \dstepping would be impractically high on our platform.  The
running times on 16 machines are shown in Table~\ref{tab:big-graphs}.
Considering that the size of {\tt R-MAT(29)} (resp., {\tt roads(32)})
is 32 (resp., about 57) times larger than the size of {\tt R-MAT(24)}
(resp., {\tt roads-USA}) the experiment shows that \cldiam's
performance scales well with the graph size on the same machine
configuration. A slight penalty in the running time on large graphs is
to be expected due to the increased number of interactions with the
disks occurring in each machine because of the large graph size.

\vspace*{0.5cm} 
Altogether, the experiments suggest that our algorithm can be
effectively employed to provide a good estimate of the diameter of
huge graphs on sufficiently large clusters of commodity processors.

\if \onepagefig0
\begin{figure}
  \centering
  \includegraphics[width=.5\columnwidth]{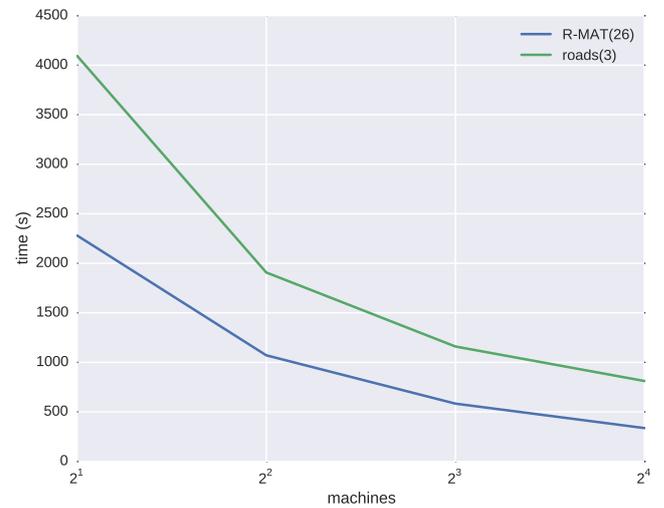}
  \caption{Scalability of \cldiam wrt the number of machines.}
  \label{fig:scalability}
\end{figure}
\else
\fi

\begin{table}
  \centering
  \begin{tabular}{l r}
    \toprule
    Graph & time (seconds) \\
    \midrule
    {\tt R-MAT(29)} & 6218 \\
    {\tt roads(32)} & 14054 \\
    \bottomrule
  \end{tabular}
  \caption{Experiments on big graphs}
  \label{tab:big-graphs}
\end{table}

\if \onepagefig1
\break
\newgeometry{textwidth=512pt,marginparwidth=0pt,marginparsep=0pt}

\begin{figure}
  \centering
  \begin{minipage}{.48\linewidth}
    \centering
    \includegraphics[width=\columnwidth]{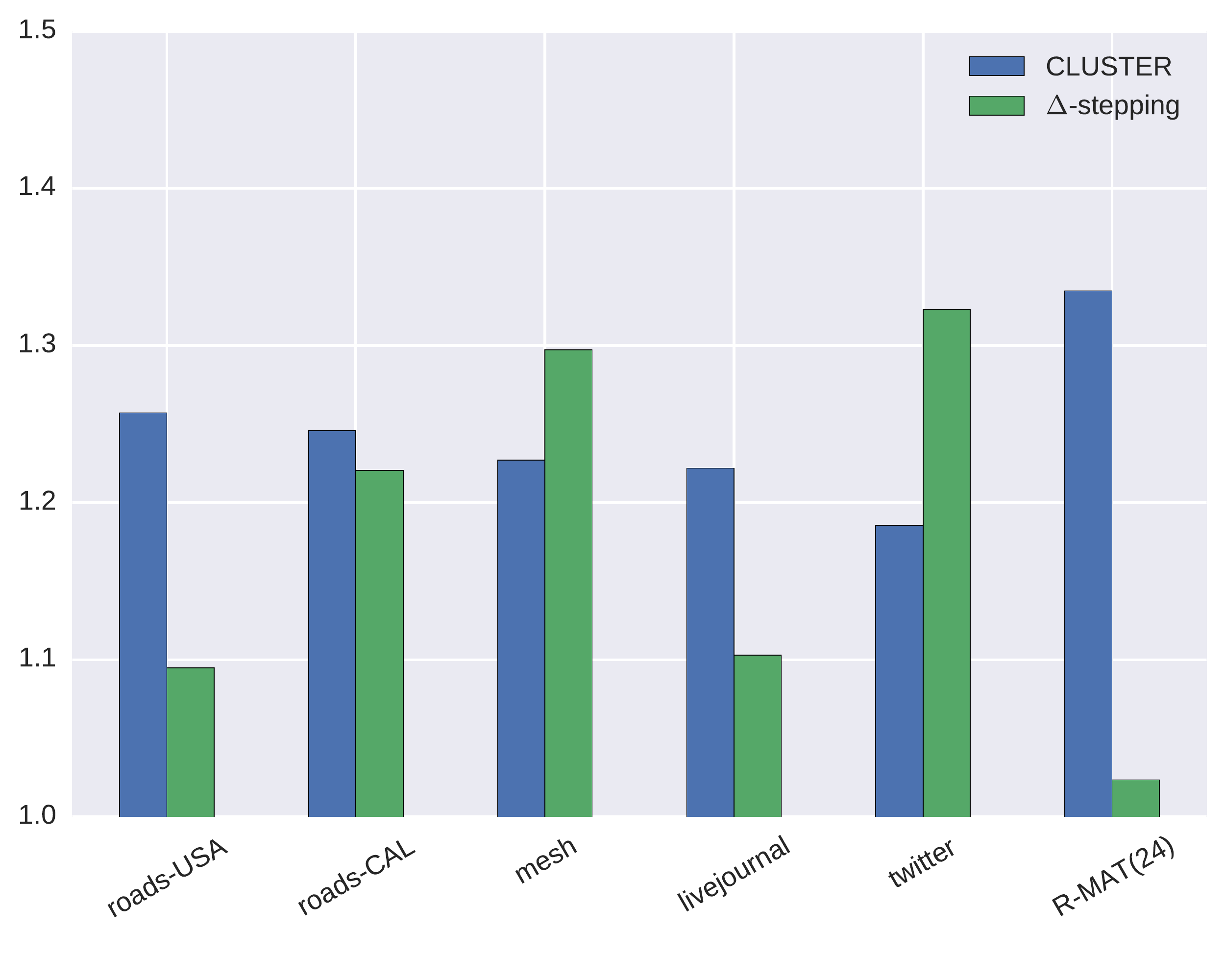}
    \caption{Approximation ratio of \cldiam and
      \dstepping.}
    \label{fig:approximation}
  \end{minipage}
  \hfill
  \begin{minipage}{.48\linewidth}
    \centering
    \includegraphics[width=\columnwidth]{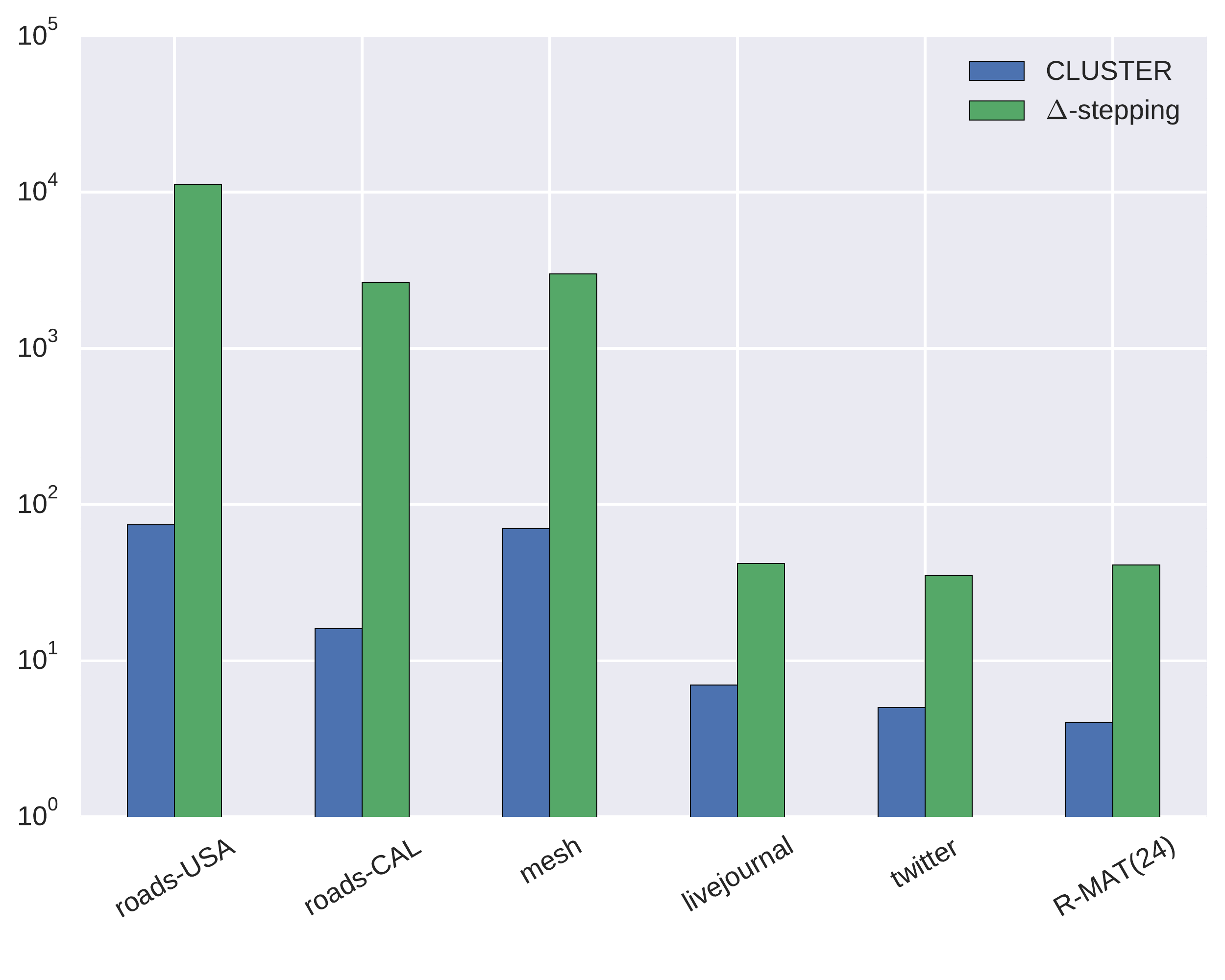}
    \caption{Number of rounds required by \cldiam and
      \dstepping. The scale is logarithmic.}
    \label{fig:num-rounds}
  \end{minipage}

  \vspace{64pt}

  \begin{minipage}{.48\linewidth}
    \centering
    \includegraphics[width=\columnwidth]{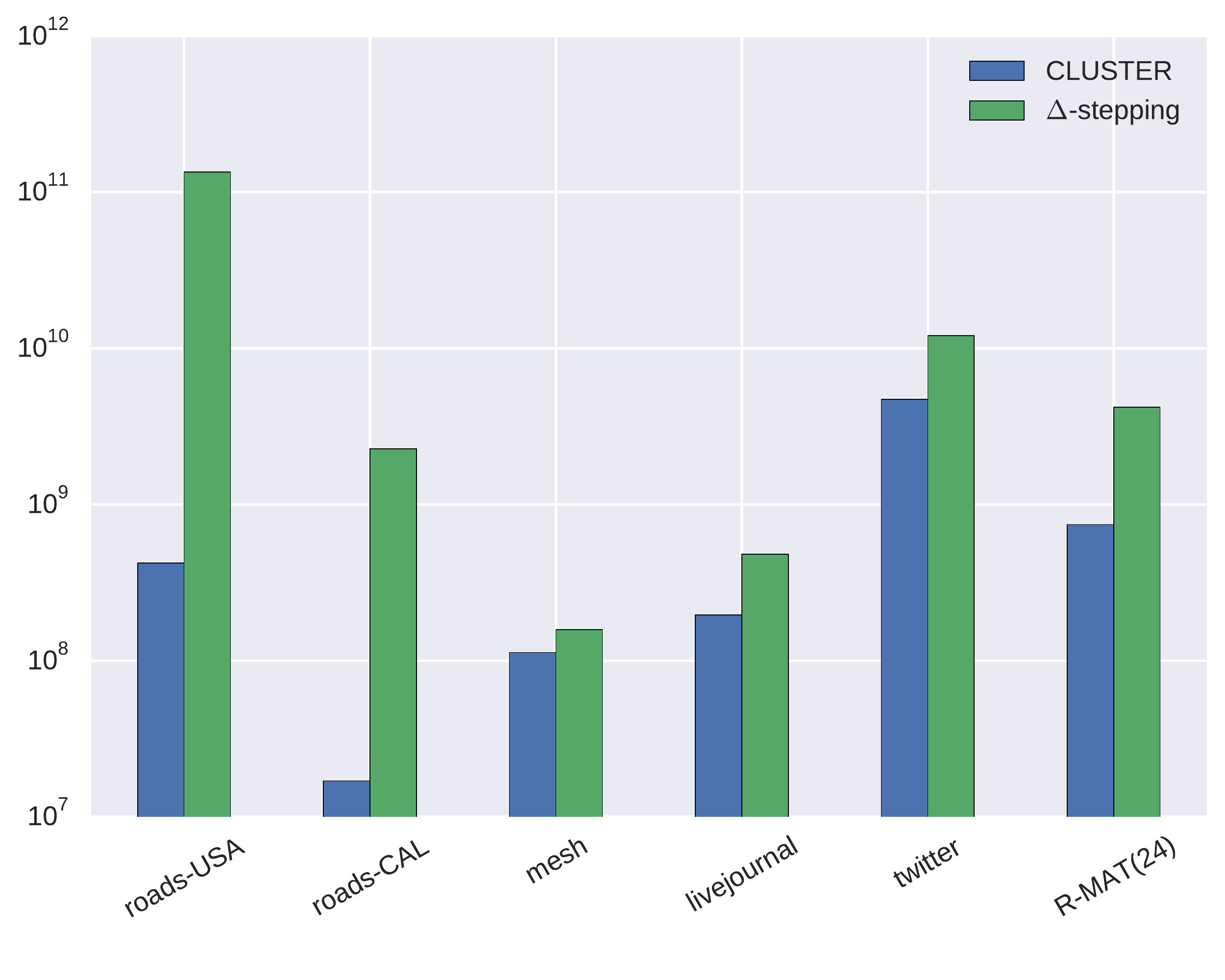}
    \caption{Work performed by \cldiam and \dstepping. The scale
      is logarithmic.}
    \label{fig:work}
  \end{minipage}
  \hfill
  \begin{minipage}{.48\linewidth}
    \centering
    \includegraphics[width=\columnwidth]{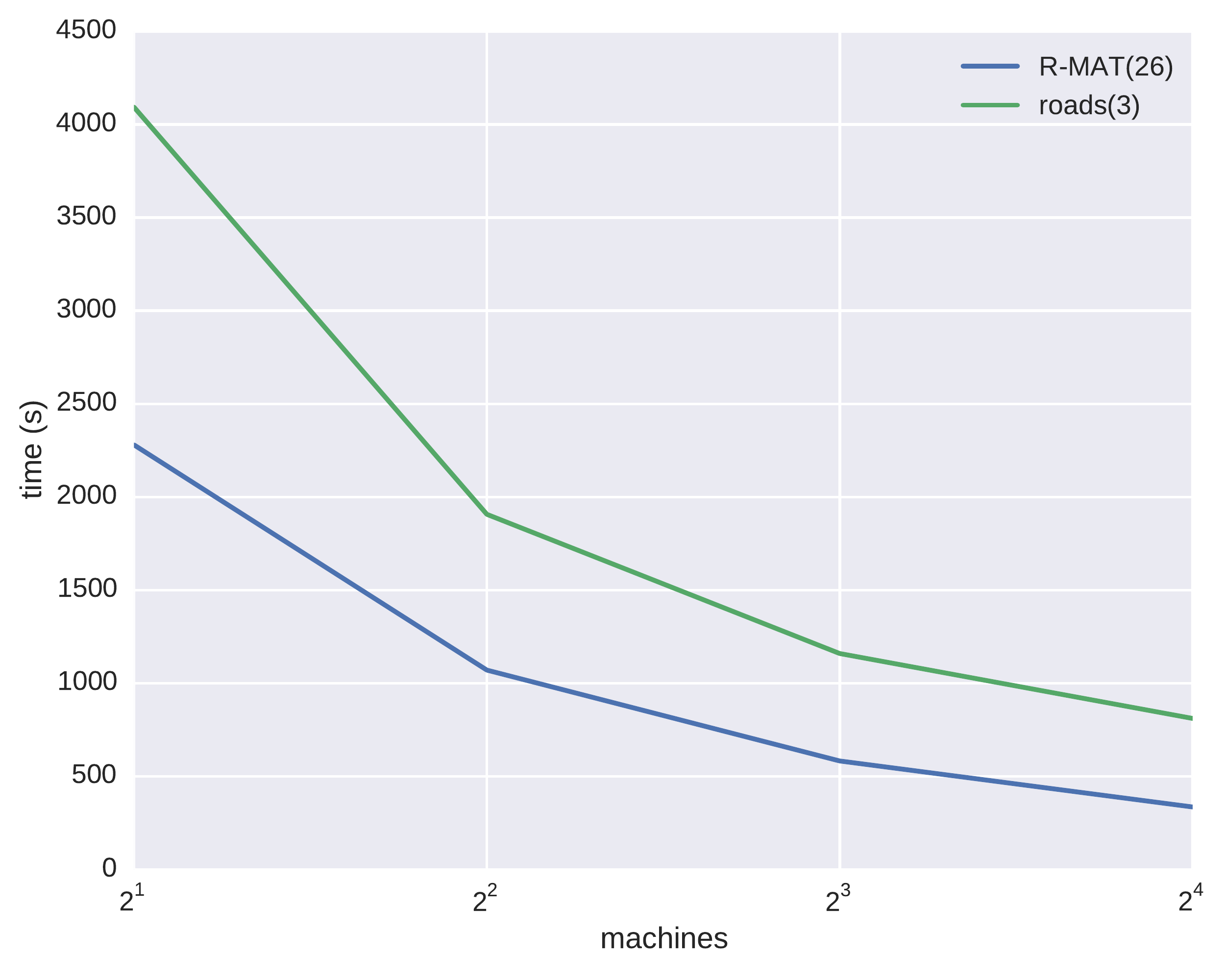}
    \caption{Scalability of \cldiam wrt the number of machines.}
    \label{fig:scalability}
  \end{minipage}
\end{figure}

\restoregeometry
\else
\fi

\if \onepagefig1
\newgeometry{top=101.40665pt}
\else
\fi


\begin{thebibliography}{AGGM06}

\bibitem[AGGM06]{AbrahamGGM2006}
I.~Abraham, C.~Gavoille, A.V. Goldberg, and D.~Malkhi.
\newblock Routing in networks with low doubling dimension.
\newblock In {\em Proc. IEEE-ICDCS}, 2006.

\bibitem[BRV11]{BoldiRV11}
P.~Boldi, M.~Rosa, and S.~Vigna.
\newblock Hyper{ANF}: approximating the neighbourhood function of very large
  graphs on a budget.
\newblock In {\em Proc. WWW}, pages 625--634, 2011.

\bibitem[CGLM12]{CrescenziGLM12}
P.~Crescenzi, R.~Grossi, L.~Lanzi, and A.~Marino.
\newblock On computing the diameter of real-world directed (weighted) graphs.
\newblock In {\em Proc. SEA}, pages 99--110, 2012.

\bibitem[CLR{\etalchar{+}}14]{ChechikLRSTW14}
S.~Chechik, D.~Larkin, L.~Roditty, G.~Schoenebeck, R.E.Tarjan, and V.V.
  Williams.
\newblock Better approximation algorithms for the graph diameter.
\newblock In {\em Proc. ACM-SIAM SODA}, pages 1041--1052, 2014.

\bibitem[CLRS09]{CormenLRS09}
T.H. Cormen, C.E. Leiserson, R.L. Rivest, and C.~Stein.
\newblock {\em Introduction to Algorithms. Third Edition}.
\newblock The MIT Press, 2009.

\bibitem[Coh00]{Cohen00}
E.~Cohen.
\newblock Polylog-time and near-linear work approximation scheme for undirected
  shortest paths.
\newblock {\em J. ACM}, 47(1):132--166, 2000.

\bibitem[CPPU15]{CeccarelloPPU15}
M.~Ceccarello, A.~Pietracaprina, G.~Pucci, and E.~Upfal.
\newblock Space and time efficient parallel graph decomposition, clustering,
  and diameter approximation.
\newblock In {\em Proc. ACM-SPAA}, pages 182--191, 2015.

\bibitem[CZF04]{Chakrabarti2004}
Deepayan Chakrabarti, Yiping Zhan, and Christos Faloutsos.
\newblock R-mat: A recursive model for graph mining.
\newblock In {\em Proc. SIAM-SDM}, volume~4, pages 442--446, 2004.

\bibitem[DG08]{DeanG08}
J.~Dean and S.~Ghemawat.
\newblock Mapreduce: simplified data processing on large clusters.
\newblock {\em CACM}, 51(1):107--113, 2008.

\bibitem[DIM]{DIMACS}
Datasets of the {9th DIMACS Implementation Challenge - Shortest Paths}.
\newblock {\tt http://www.dis.uniroma1.it/challenge9/}.

\bibitem[Dwa69]{Dwass69}
M.~Dwass.
\newblock The total progeny in a branching process and a related random walk.
\newblock {\em Journal of Applied Probability}, 6(3):682--686, 1969.

\bibitem[GSZ11]{GoodrichSZ11}
M.T. Goodrich, N.~Sitchinava, and Q.~Zhang.
\newblock Sorting, searching, and simulation in the {MapReduce} framework.
\newblock In {\em Proc. ISAAC}, pages 374--383, 2011.

\bibitem[KS97]{klein_randomized_1997}
P.~Klein and S.~Subramanian.
\newblock A randomized parallel algorithm for single-source shortest paths.
\newblock {\em Journal of Algorithms}, 25(2):205--220, 1997.

\bibitem[KSV10]{KarloffSV10}
H.~Karloff, S.~Suri, and S.~Vassilvitskii.
\newblock A model of computation for mapreduce.
\newblock In {\em Proc. ACM-SIAM SODA}, pages 938--948, 2010.

\bibitem[LAW]{LAW}
{Laboratory for web algorithmics datasets}.
\newblock {\tt http://law.di.unimi.it/webdata/twitter-2010/}.

\bibitem[Mey08]{Meyer08}
U.~Meyer.
\newblock On trade-offs in external-memory diameter-approximation.
\newblock In {\em Proc. SWAT}, pages 426--436, 2008.

\bibitem[MS03]{meyer2003delta}
U.~Meyer and P.~Sanders.
\newblock $\delta$-stepping: a parallelizable shortest path algorithm.
\newblock {\em Journal of Algorithms}, 49(1):114--152, 2003.

\bibitem[PPR{\etalchar{+}}12]{PietracaprinaPRSU12}
A.~Pietracaprina, G.~Pucci, M.~Riondato, F.~Silvestri, and E.~Upfal.
\newblock Space-round tradeoffs for {MapReduce} computations.
\newblock In {\em Proc. ACM-ICS}, pages 235--244, 2012.

\bibitem[SNA]{SNAP}
{Stanford Large Network Dataset Collection}.
\newblock {\tt http://snap.stanford.edu/data}.

\bibitem[SPA]{SPARK}
{Spark: Lightning-fast cluster computing.}
\newblock {\tt http://spark.apache.org}.

\end{thebibliography}

\newcommand{\etalchar}[1]{$^{#1}$}

\end{document}